\documentclass[conference]{IEEEtran}

\usepackage{perlaza}
\usepackage{tikz}

\makeatletter
\renewcommand{\baselinestretch}{.93}
\IEEEoverridecommandlockouts
\begin{document}
\title{Simultaneous Information and Energy Transmission with Finite Constellations} 

\author{Sadaf ul Zuhra, Samir M. Perlaza, and Eitan Altman
		
		\thanks{Sadaf ul Zuhra, Samir M. Perlaza, and Eitan Altman are with INRIA, Centre de Recherche de Sophia Antipolis - M\'{e}diterran\'{e}e, 2004  Route des Lucioles, 06902 Sophia Antipolis, France. $\lbrace$sadaf-ul.zuhra, samir.perlaza, eitan.altman$\rbrace$@inria.fr}
		\thanks{Samir M. Perlaza is also with the Electrical and Computer Engineering Department, Princeton University, Princeton, 08544 NJ, USA; and the Laboratoire de Math\'{e}matiques GAATI, Universit\'{e} de la Polyn\'{e}sie Fran\c{c}aise,  BP 6570, 98702 Faaa, French Polynesia.}
		\thanks{Eitan Altman is also with the Laboratoire d'Informatique d'Avignon (LIA), Universit\'{e} d'Avignon, 84911 Avignon, France; and with the Laboratory of Information, Network and Communication Sciences (LINCS), 75013 Paris, France.}
		\thanks{This research was supported in part by the European Commission through the H2020-MSCA-RISE-2019 program under grant 872172; in part by the Agence Nationale de la Recherche (ANR) through the project MAESTRO-5G (ANR-18-CE25-0012); and in part by the Fondation Mathématique Jacques Hadamard (FMJH) through the Programme Gaspard Monge.}
	}

\maketitle

\begin{abstract}
In this paper, the fundamental limits on the rates at which information and energy can be simultaneously transmitted over an additive white Gaussian noise channel are studied under the following assumptions: $(a)$ the channel is memoryless; $(b)$ the number of channel input symbols  (constellation size) and block length are finite; and $(c)$ the decoding error probability (DEP) and the energy outage probability (EOP) are bounded away from zero.  In particular, it is shown that the  limits on the maximum  information and energy transmission rates; and the minimum DEP and EOP, are essentially set by the type induced by the code used to perform the transmission. That is, the empirical frequency with which each channel input symbol appears in the codewords.  Using this observation, guidelines for optimal constellation design for simultaneous energy and information transmission are presented.
\end{abstract}

\section{Introduction} \label{sec:intro}

Nikola Tesla suggested that radio frequency signals can be used to simultaneously transmit  both information and energy in 1914 \cite{Tesla-Patent-1914}.  About two centuries later, simultaneous information and energy transmission (SIET), also known as simultaneous wireless information and power transfer (SWIPT) is one of the technologies that might be implemented in 6G communication systems in the near future, c.f., \cite{6GSpeculation} and\cite{6GSpeculationA}. SIET implies a fundamental trade-off between the amount of energy and information that can be simultaneously transmitted by a signal. This has been the subject of intense research, c.f.,~\cite{varshney2008transporting, amor2016fundamental, perlaza2018simultaneous, survey} and \cite{surveyA}. 
In \cite{varshney2008transporting}, a capacity-energy function is defined in order to determine the fundamental limit on the information transmission rate subject to the fact that the average energy at the channel output is not smaller than a given threshold. Therein, the underlying assumption is that  the communication duration in channel uses is infinitely long. This guarantees that the decoding error probability (DEP) and the energy outage probability (EOP) can be made arbitrarily close to zero, and thus, the focus is only on the information transmission rate and the energy transmission rate. This analysis has been extended to multi-user channels. In this case, the notion of information-energy capacity region generalizes to the set of all information and energy rate tuples that can be simultaneously achieved in the asymptotic block length regime \cite{amor2016fundamental}. For instance, the information-energy capacity region of the Gaussian multiple access channel is characterized in \cite{amor2016feedback}, whereas the information-energy capacity region of the Gaussian interference channel is approximated in  \cite{KhalfetGIC}.

A first attempt to study the fundamental limits of SIET under the assumption of finite  transmission duration with finite constellation sizes is presented in \cite{perlaza2018simultaneous}. Therein, the study is restricted to discrete memoryless channels and it is shown that finite transmission duration implies DEPs and EOPs that are bounded away from zero. 

This work contributes in this direction and considers the problem of SIET in additive white Gaussian noise channels considering finite transmission duration and finite constellation sizes.   
The main results in this work highlight the intuition that codes that uniformly use all channel input symbols are associated with high information rates, whereas, codes that exclusively use the channel input symbols that carry the largest amount of energy are associated with high energy rates. 
More specifically,  a characterization of the maximum information and energy rates and minimum DEP and EOP that can be simultaneously achieved is formulated  in terms of the type  the code induces on the set of channel input symbols. In this work, a type is understood in the sense of  the empirical frequency with which each channel input symbol appears in the codewords \cite{CsiszarMoT}.  
  
\section{System Model} \label{sec:system_model}
Consider a communication system formed by a transmitter, an information receiver (IR), and an energy harvester (EH). The objective of the transmitter is to simultaneously send information to the IR at a rate of $R$ bits per second; and energy to the EH at a rate of $B$ Joules per second over an additive white Gaussian noise (AWGN) channel. 
That is, given a channel input
${\boldsymbol x} = (x_1,x_2, \ldots, x_n)^{\sf{T}} \in \mathds{C}^{n}$, with $n \in \mathds{N}$, the outputs of the channel are the random vectors
\begin{subequations}\label{EqChannelModel}
\begin{IEEEeqnarray}{rcl}
    \label{eq:channel1}
    {\boldsymbol Y} & = & {\boldsymbol x} + {\boldsymbol N}_1, \mbox{ and }  \\
    \label{eq:channel2} 
    {\boldsymbol Z} & = & {\boldsymbol x} + {\boldsymbol N}_2,
\end{IEEEeqnarray}
\end{subequations}
where $n$ is the duration of the transmission in channel uses; and the vectors ${\boldsymbol Y} = (Y_1,Y_2, \ldots, Y_n)^{\sf{T}} \in \mathds{C}^{n}$ and ${\boldsymbol Z} = (Z_1,Z_2, \ldots, Z_n)^{\sf{T}} \in \mathds{C}^{n}$ are the inputs of the IR and the EH, respectively. The components of the random vectors ${\boldsymbol N}_1 = (N_{1,1}, N_{1,2}, \ldots, N_{1,n})^{\sf{T}}\in \mathds{C}^{n}$ and ${\boldsymbol N}_2 = (N_{2,1}, N_{2,2}, \ldots, N_{2,n})^{\sf{T}}\in \mathds{C}^{n}$ are independent and identically distributed. More specifically, for all $(i,j) \in \{1,2\} \times \{1,2, \ldots, n\}$,  $N_{i,j}$ is a complex circularly symmetric Gaussian random variable whose real and imaginary parts have zero means and variances $\frac{1}{2}\sigma^2$. 
\par
That is, for all $\boldsymbol{y} = (y_1, y_2, \ldots, y_n)^{\sf{T}} \in \mathds{C}^{n}$, for all $\boldsymbol{z} = (z_1, z_2, \ldots, z_n)^{\sf{T}} \in \mathds{C}^{n}$, and for all $\boldsymbol{x} = (x_1, x_2, \ldots, x_n)^{\sf{T}} \in \mathds{C}^{n}$, it holds that the joint probability density function of the channel outputs $(\boldsymbol{Y}, \boldsymbol{Z})$ satisfies $f_{\boldsymbol{Y Z}|\boldsymbol{X}}(\boldsymbol{y}, \boldsymbol{z}|\boldsymbol{x}) = f_{\boldsymbol{Y}|\boldsymbol{X}}(\boldsymbol{y}|\boldsymbol{x}) f_{\boldsymbol{Z}|\boldsymbol{X}}(\boldsymbol{z}|\boldsymbol{x})$, where  
\begin{IEEEeqnarray}{rcl} \label{EqYXdistribution}
    f_{\boldsymbol{Y}|\boldsymbol{X}}(\boldsymbol{y}|\boldsymbol{x}) & = & \prod_{t=1}^n f_{Y|X}(y_t|x_t) \mbox{ and }\\
    \label{eq:z_distribution}
    f_{\boldsymbol{Z}|\boldsymbol{X}}(\boldsymbol{z}|\boldsymbol{x}) & = & \prod_{t=1}^n f_{Z|X}(z_t|x_t),
\end{IEEEeqnarray}
and for all $t \in \lbrace 1,2, \ldots, n \rbrace$,
\begin{subequations}\label{EqDensities}
\begin{IEEEeqnarray}{l}
\nonumber
    f_{Y|X}(y_t|x_t) = \\
\label{EqYXdistribution2}    
    \frac{1}{\pi \sigma^2} \exp \left(- \frac{(\Re(y_t)-\Re(x_t))^2 +(\Im(y_t) - \Im(x_t))^2}{\sigma^2} \right), \quad \\
\nonumber
    f_{Z|X}(z_t|x_t) = \\
\label{eq:z_distribution2}    
    \frac{1}{\pi \sigma^2} \exp \left(- \frac{(\Re(z_t)-\Re(x_t))^2 +(\Im(z_t) - \Im(x_t))^2}{\sigma^2} \right).\quad
\end{IEEEeqnarray}
\end{subequations}
Within this framework, two tasks must be accomplished: information transmission and energy transmission. 

\subsection{Information Transmission} \label{subsec:information_transmission}
Assume that the information transmission takes place using a modulation scheme that uses $L$ symbols. That is, there is a set
\begin{equation}\label{EqCIsymbols}
\mathcal{X} \triangleq \{x^{(1)}, x^{(2)}, \ldots, x^{(L)}\} \subset \mathds{C}
\end{equation} that contains all possible channel input symbols, and  $L \triangleq \left| \mathcal{X} \right|$. 
Let $M$ be the number of message indices to be transmitted within $n$ channel uses. That is, $M \leqslant 2^{n\left\lfloor \log_2 L \right\rfloor}$ 
To reliably transmit a message index, the transmitter uses an $(n,M)$-code defined as follows.
\begin{definition} \label{def:nm_code}
$(n,M)$-code: An $(n,M)$-code for the random transformation in~\eqref{EqChannelModel} is a system:
\begin{equation}
    \{({\boldsymbol u}(1),\mathcal{D}_1), ({\boldsymbol u}(2),\mathcal{D}_2), \ldots, ({\boldsymbol u}(M),\mathcal{D}_M)\},
\end{equation}
where, for all $(i,j) \in \{1,2, \ldots, M\}^2, i\neq j$,
\begin{subequations}\label{EqCodeProperties}
\begin{align}
\label{eq:u_i}  &{\boldsymbol u}(i) = (u_1(i), u_2(i), \ldots, u_n(i)) \in \mathcal{X}^n, \\
        &\mathcal{D}_i \cap \mathcal{D}_j = \phi,\\
        &\bigcup_{i = 1}^M \mathcal{D}_i \subseteq \mathds{C}^n, \mbox{ and }\\
  \label{eq:P_criteria}      &|u_t(i)| \leqslant P,
    \end{align}
\end{subequations}
\end{definition}
where $P$ is the peak-power constraint.
Assume that the transmitter uses the $(n,M)$-code
\begin{equation} \label{Eqnm_code}
    \mathscr{C} \triangleq \{({\boldsymbol u}(1),\mathcal{D}_1), ({\boldsymbol u}(2),\mathcal{D}_2), \ldots, ({\boldsymbol u}(M),\mathcal{D}_M)\},
\end{equation}
that satisfies~\eqref{EqCodeProperties}.
The information rate of any  $(n,M)$-code is given by
\begin{equation} \label{EqR}
    R = \frac{\log_2 M}{n}
\end{equation}
in bits per channel use.
To transmit the message index $i$, with $i \in \lbrace 1,2, \ldots, M \rbrace$, the transmitter uses the codeword ${\boldsymbol u}(i)= (u_1(i), u_2(i), \ldots, u_n(i))$. That is, at channel use $t$, with $t \in \lbrace 1,2, \ldots, n\rbrace$, the transmitter inputs the symbol $u_{t}(i)$ into the channel. At the end of $n$ channel uses, the IR observes a realization of the random vector ${\boldsymbol Y} = (Y_1, Y_2, \ldots, Y_n)^{\sf{T}}$ in~\eqref{eq:channel1}. The IR decides that message index $j$, with $j \in \lbrace 1,2, \ldots, M \rbrace$, was transmitted, if the event ${\boldsymbol Y} \in \mathcal{D}_j$ takes place, with $\mathcal{D}_j$ in~\eqref{Eqnm_code}.
That is, the set $\mathcal{D}_j \in \mathds{C}^n$ is the region of correct detection for message index $j$.
Therefore, the DEP associated with the transmission of message index $i$ is
\begin{IEEEeqnarray}{rcl}
\label{EqDEPi}
    \gamma_i(\mathscr{C}) %&=& \mathrm{Pr}[{\boldsymbol  Y} \in \mathcal{D}_i^{\sf{c}}|{\boldsymbol X} = {\boldsymbol  u}(i)], %\\
%    &= \sum_{j \neq i} \mathrm{Pr}[{\boldsymbol  Y} \in \mathcal{D}_j|{\boldsymbol X} = {\boldsymbol  u}(i)], \\
    &\triangleq& 1 - \int_{\mathcal{D}_i} f_{\boldsymbol{Y}|\boldsymbol{X}}(\boldsymbol{y}|\boldsymbol{u}(i)) \mathrm{d}\boldsymbol{y},
\end{IEEEeqnarray}
and the average DEP is
\begin{IEEEeqnarray}{rcl} 
    \label{eq:gamma} 
    \gamma(\mathscr{C}) &\triangleq& \frac{1}{M}\sum_{i = 1}^M \gamma_i(\mathscr{C}). %\\
%    &=& \frac{1}{M}\sum_{i = 1}^M \left( 1 - \int_{\mathcal{D}_i} f_{\boldsymbol{Y}|\boldsymbol{X}}(\boldsymbol{y}|\boldsymbol{u}(i)) \mathrm{d}y \right),
\end{IEEEeqnarray}

Using this notation, Definition~\ref{def:nm_code} can be refined as follows.
\begin{definition}[$(n,M,\epsilon)$-codes] \label{def:nme_code}
An $(n,M)$-code for the random transformation in~(\ref{EqChannelModel}), denoted by $\mathscr{C}$, is said to be an $(n,M,\epsilon)$-code  if 
\begin{equation}\label{eq:gamma_upperbound}
    \gamma(\mathscr{C}) < \epsilon.
\end{equation}
\end{definition}
\subsection{Energy Transmission} \label{subsec:energy_transmission}
Let $g:\mathds{C}\rightarrow [0, +\infty]$ be a positive function such that given a channel output $z \in \mathds{C}$, the value $g(z)$ is the energy harvested from such channel output. 
The energy transmission task must ensure that a minimum average energy $B$ is harvested at the EH at the end of $n$ channel uses. 
Let $\bar{g}:\mathds{C}^n \rightarrow [0, +\infty]$ be a positive function such that given $n$ channel outputs ${\boldsymbol z} = (z_1,z_2, \ldots, z_n)$, the average energy is
\begin{equation}
    \bar{g}({\boldsymbol z}) = \frac{1}{n}\sum_{t=1}^n g(z_t),
\end{equation}
in energy units per channel use.
Assume that the transmitter uses the code $\mathscr{C}$ in~\eqref{Eqnm_code}. Then, the EOP associated with the transmission of message index $i$, with $i \in \lbrace 1,2, \ldots, M \rbrace$, is 
\begin{IEEEeqnarray}{rcl}
\label{Eqthetai}
\theta_i(\mathscr{C},B) & \triangleq & \mathrm{Pr}[\bar{g}({\boldsymbol Z}) < B |{\boldsymbol X} = {\boldsymbol u}(i)], 
\end{IEEEeqnarray}
where the probability is with respect to the probability density function $f_{\boldsymbol{Z}|\boldsymbol{X}}$ in~\eqref{eq:z_distribution}; and the average EOP is given by
\begin{IEEEeqnarray}{rcl} 
\label{eq:theta_def}
    \theta(\mathscr{C},B) & \triangleq & \frac{1}{M}\sum_{i=1}^M \theta_i(\mathscr{C},B). 
\end{IEEEeqnarray}
This leads to the following refinement of Definition~\ref{def:nme_code}.
\begin{definition}[$(n,M,\epsilon,B,\delta)$-code] \label{def:nmed_code}
An $(n,M,\epsilon)$-code for the random transformation in~\eqref{EqChannelModel}, denoted by $\mathscr{C}$, is said to be an $(n,M,\epsilon,B,\delta)$-code  if 
\begin{equation} \label{eq:delta}
    \theta(\mathscr{C},B) < \delta.
\end{equation}
\end{definition}
%
% {\color{blue}
% \begin{definition}
% Rate optimality: A code $\mathscr{C}$ of the form in~\eqref{Eqnm_code} is said to be a rate optimal code if 
% \begin{equation}
%     \sum_{\ell = 1}^L \mathds{1}_{\{P_{\mathscr{C}}(x^{(\ell)})>0\}} = L.
% \end{equation}
% \end{definition}
% %
% \begin{definition}
% Energy optimality: A code $\mathscr{C}$ of the form in~\eqref{Eqnm_code} is said to be a $B$-energy optimal code if
% \begin{equation}
%     \theta(\mathscr{C},B) = 0.
% \end{equation}
% \end{definition}
% }
\section{Main Results} \label{sec:results}

The results in this section are presented in terms of the types induced by the codewords of a given code.
Given an $(n,M,\epsilon,B,\delta)$-code denoted by $\mathscr{C}$ of the form in~\eqref{Eqnm_code}, the type induced by the codeword $\boldsymbol{u}(i)$, with $i \inCountK{M}$, is a probability mass function (pmf) whose support is $\mathcal{X}$ in~\eqref{EqCIsymbols}. Such pmf is denoted by $P_{\boldsymbol{u}(i)}$ and for all $\ell \in \lbrace 1,2, \ldots, L \rbrace$,
\begin{equation} \label{eq:u_measure}
    P_{\boldsymbol{u}(i)}(x^{(\ell)}) \triangleq \frac{1}{n} \sum_{t=1}^n \mathds{1}_{\{u_t(i) = x^{(\ell)}\}},
\end{equation}
where $x^{(\ell)}$ is an element of $\mathcal{X}$ in~\eqref{EqCIsymbols}.
The type induced by all the codewords in $\mathscr{C}$ is also a pmf on the set $\mathcal{X}$ in~\eqref{EqCIsymbols}. Such pmf is denoted by  $P_{\mathscr{C}}$ and  for all $\ell \in \lbrace 1,2, \ldots, L \rbrace$,
\begin{equation} \label{eq:p_bar}
    P_{\mathscr{C}}(x^{(\ell)}) \triangleq \frac{1}{M} \sum_{i=1}^M P_{\boldsymbol{u}(i)}(x^{(\ell)}).
\end{equation}

Using this notation, the main results are essentially upper bounds on the information rate $R$ in~\eqref{EqR} and energy  rate $B$ in 
\eqref{Eqthetai}; as well as, lower bounds on the average DEP $\epsilon$ and average EOP $\delta$ for all possible $(n,M,\epsilon,B,\delta)$-codes. These bounds are provided assuming that the $L$ channel input symbols in $\mathcal{X}$ and the block length $n$ are finite.

\subsection{Energy Transmission Rate and Average EOP}

The following lemma introduces an upper bound on the energy transmission rate that holds for all possible $(n,M,\epsilon,B,\delta)$-codes. 

\begin{lemma} \label{LemmaB}
Given an $(n,M,\epsilon,B,\delta)$-code denoted by $\mathscr{C}$ for the random transformation in~\eqref{EqChannelModel} of the form in~\eqref{Eqnm_code}, the energy transmission rate $B$ satisfies,
\begin{equation} \label{eq:B_bound_lemma}
    B \leq \frac{1}{1-\delta} \sum_{\ell=1}^L P_{\mathscr{C}}(x^{(\ell)}) \mathbb{E}\left[g\left(x^{(\ell)} + W \right) \right],
\end{equation}
where $P_{\mathscr{C}}$ is defined in~\eqref{eq:p_bar} and the expectation is with respect to $W$, which  is a complex circularly symmetric Gaussian random variable whose real and imaginary parts have zero means and variances $\frac{1}{2}\sigma^2$.
\end{lemma}
\begin{IEEEproof} 
The proof of Lemma \ref{LemmaB} is presented in \cite{Zuhra-Inria-TR}.
\end{IEEEproof}
An observation from Lemma~\ref{LemmaB} is that the limit on the energy rate $B$ does not depend on the individual types $P_{\boldsymbol{u}(1)}$,  $P_{\boldsymbol{u}(2)}$, $\ldots$,  $P_{\boldsymbol{u}(M)}$ but on the type induced by all codewords, i.e.,  $P_{\mathscr{C}}$.
The inequality in~\eqref{eq:B_bound_lemma} implies the following corollary. 

\begin{corollary}\label{CorDelta}
Given an $(n,M,\epsilon,B,\delta)$-code denoted by $\mathscr{C}$ for the random transformation in~\eqref{EqChannelModel} of the form in~\eqref{Eqnm_code}, the average EOP  $\delta$ satisfies, \begin{equation}
    \delta \geq \left(1 - \frac{1}{B} \sum_{\ell=1}^L P_{\mathscr{C}}(x^{(\ell)}) \mathbb{E} \left[g(x^{(\ell)} + W \right]\right)^+.
\end{equation}
where $P_{\mathscr{C}}$ is defined in~\eqref{eq:p_bar} and the expectation is with respect to $W$, which  is a complex circularly symmetric Gaussian random variable whose real and imaginary parts have zero means and variances $\frac{1}{2}\sigma^2$.
 \end{corollary}

\subsection{Average Decoding Error Probability}
The analysis of the information transmission rate and the average DEP of a given code $\mathscr{C}$ of the form in~\eqref{Eqnm_code} depends on the choice of the decoding sets $\mathcal{D}_1$, $\mathcal{D}_2$, $\ldots$, $\mathcal{D}_M$. 
Without any loss of generality, assume that for all $i \inCountK{M}$, the decoding set $\mathcal{D}_i$ is written in the form
\begin{equation}\label{EqDit}
    \mathcal{D}_i = \mathcal{D}_{i,1} \times \mathcal{D}_{i,2} \times \ldots \times \mathcal{D}_{i,n},
\end{equation}
where for all $t \inCountK{n}$, the set $\mathcal{D}_{i,t}$ is a subset of~$\complex$.

The following lemma introduces a lower bound on the average DEP that holds for all possible $(n,M,\epsilon,B,\delta)$-codes. 

\begin{lemma} \label{LemmaEpsilonLowerBound}
Given an $(n,M,\epsilon,B,\delta)$-code $\mathscr{C}$ for the random transformation in~\eqref{EqChannelModel} of the form in~\eqref{Eqnm_code}, 
the average DEP  $\epsilon$ satisfies,
\begin{IEEEeqnarray}{rcl}
\nonumber
\epsilon & \geq & 1 - \frac{1}{M} \sum_{i=1}^M \exp \Bigg( -n H(P_{\boldsymbol{u}(i)}) - n D(P_{\boldsymbol{u}(i)} || Q) \\
\label{EqGammaLBc}
  & & + n\log \sum_{j=1}^L \int_{\mathcal{E}_{j}} f_{Y|X}(y|x^{(j)}) \mathrm{d}y  \Bigg),
\end{IEEEeqnarray}
where, $P_{\boldsymbol{u}(i)}$ is the type defined in~\eqref{eq:u_measure}; the function $Q$ is a pmf on $\mathcal{X}$ in~\eqref{EqCIsymbols} such that for all $i \inCountK{L}$,
\begin{equation}
\label{EqQ}
Q\left( x^{(i)} \right) = \frac{  \int_{\mathcal{E}_{i}} f_{Y|X}(y|x^{(i)}) \mathrm{d}y }{\sum_{j=1}^L  \int_{\mathcal{E}_{j}} f_{Y|X}(y|x^{(j)}) \mathrm{d}y };
\end{equation}
and
for all $\ell \inCountK{L}$,  the set $\mathcal{E}_{\ell}$ is
\begin{equation}
\mathcal{E}_{\ell} = \mathcal{D}_{i^\star,t^\star},
\end{equation}
where $(i^\star,t^\star)$ satisfies
\begin{equation}\label{EqEell}
(i^\star,t^\star) \triangleq \arg\min_{(i,t) \in \CountK{M}\times\CountK{n} } \int_{\mathcal{D}_{i,t}} f_{Y|X}(y|x^{(\ell)}) \mathrm{d}y.
\end{equation}
\end{lemma}
\begin{IEEEproof} 
The proof of Lemma \ref{LemmaEpsilonLowerBound} is presented in \cite{Zuhra-Inria-TR}.
\end{IEEEproof}
 
A class of codes that is of particular interest in this study is that of  homogeneous codes, which are defined hereunder.

\begin{definition}[Homogeneous Codes]\label{DefHC}
An $(n,M,\epsilon,B,\delta)$-code denoted by $\mathscr{C}$ for the random transformation in~\eqref{EqChannelModel} of the form in~\eqref{Eqnm_code} is said to be homogeneous if for all $i \inCountK{M}$ and for all $\ell \inCountK{L}$, it holds that
\begin{equation}\label{EqHomogeneousCodes}
 P_{\boldsymbol{u}(i)}(x^{(\ell)}) =  P_{\mathscr{C}}(x^{(\ell)}),
\end{equation}
and for all $(i,t) \inCountK{M}\times \CountK{n}$ for which $u_{t}(i) = x^{(\ell)}$, for some $\ell \inCountK{L}$, it holds that,
\begin{equation}
\label{EqDitHC}
\mathcal{D}_{i,t} = \mathcal{E}_{\ell},
\end{equation}
where, $P_{\boldsymbol{u}(i)}$ and $P_{\mathscr{C}}$ are the types defined in~\eqref{eq:u_measure} and~\eqref{eq:p_bar}, respectively;  the sets $\mathcal{D}_{i,t}$ are defined in~\eqref{EqDit}; and the sets $\mathcal{E}_{\ell}$ are defined in~\eqref{EqEell}.
\end{definition}

Homogeneous codes are essentially $(n,M,\epsilon,B,\delta)$-codes that satisfy two conditions. First,  a given channel input symbol is used the same number of times in all codewords; and second, every channel input symbol is decoded with the same decoding set independently of the codeword and/or the position in the codeword.
Lemma~\ref{LemmaEpsilonLowerBound} simplifies for the case of homogeneous codes (Defintion~\ref{DefHC}) as follows.

\begin{corollary} \label{CorEpsilonHomogeneous}
Given a homogeneous $(n,M,\epsilon,B,\delta)$-code denoted by $\mathscr{C}$ for the random transformation in~\eqref{EqChannelModel} of the form in~\eqref{Eqnm_code}, the DEP $\epsilon$ satisfies that
\begin{IEEEeqnarray}{rcl}
\nonumber
    \epsilon & \geq & 1 - \exp \Big( -n H(P_{\mathscr{C}}) - n D(P_{\mathscr{C}} || Q) \\
\label{EqCorollary3}    
    & & + n\log \sum_{j=1}^L\int_{\mathcal{E}_{j}} f_{Y|X}(y|x^{(j)}) \mathrm{d}y  \Big),
\end{IEEEeqnarray}
where, $P_{\mathscr{C}}$ is the type defined in~\eqref{eq:p_bar}; the pmf $Q$ is defined in~\eqref{EqQ} and for all $j \inCountK{L}$,  the sets $\mathcal{E}_{j}$ are defined in~\eqref{EqEell}. 
\end{corollary}

Note that the right-hand side of~\eqref{EqCorollary3} is minimized when the following condition is met,
\begin{equation}
\label{EqCondition1}
P_{\mathscr{C}} = Q.
\end{equation} 
This observation leads to the following corollary.
\begin{corollary} \label{CorEpsilonHomogeneous2}
Given a homogeneous $(n,M,\epsilon,B,\delta)$-code denoted by $\mathscr{C}$ for the random transformation in~\eqref{EqChannelModel} of the form in~\eqref{Eqnm_code}, such that ~\eqref{EqCondition1} is satisfied, the DEP $\epsilon$ satisfies that
\begin{IEEEeqnarray}{rcl}
    \epsilon & \geqslant & 1 - \exp \left( \hspace{-2mm} -n H(Q) + n\log \hspace{-1mm} \left( \sum_{j=1}^L\int_{\mathcal{E}_{j}}  \hspace{-2mm}f_{Y|X}(y|x^{(j)}) \mathrm{d}y \hspace{-2mm} \right) \hspace{-2mm} \right), \IEEEeqnarraynumspace 
\end{IEEEeqnarray}
where the pmf $Q$ is defined in~\eqref{EqQ} and for all $j \inCountK{L}$,  the sets $\mathcal{E}_{j}$ are defined in~\eqref{EqEell}. 
\end{corollary}

Note that the equality in~\eqref{EqCondition1} reveals the optimal choice of $P_{\mathscr{C}}$ with respect to the average DEP, which is using the channel input symbols in $\mathcal{X}$ with a probability that is proportional to the probability of correct decoding. 
 \subsection{Information Transmission Rate}
 
A first upper bound on the information rate is obtained by upper bounding the number of codewords that a code might possess given the particular types $P_{\boldsymbol{u}(1)}$, $P_{\boldsymbol{u}(2)}$, $\ldots$, $P_{\boldsymbol{u}(n)}$; or the average type  $P_{\mathscr{C}}$ in~\eqref{eq:p_bar}.
The following lemma introduces such a bound for the case of a homogeneous code.

\begin{lemma} \label{lemma:R_upperbound}
Given a homogeneous $(n,M,\epsilon,B,\delta)$-code denoted by $\mathscr{C}$ for the random transformation in~\eqref{EqChannelModel} of the form in~\eqref{Eqnm_code}, the  information transmission rate $R$ in~\eqref{EqR} is such that
\begin{equation}
\label{EqRbound}
R \leq \frac{1}{n} \log_2 \left( \frac{n!}{\prod_{\ell=1}^L (nP_{\mathscr{C}}(x^{(\ell)}))!}\right)  \leq \log_2 L,\end{equation}
where, $P_{\mathscr{C}}$ is the type defined in~\eqref{eq:p_bar}.
\end{lemma}
\begin{IEEEproof} 
The proof of Lemma \ref{lemma:R_upperbound} is presented in \cite{Zuhra-Inria-TR}.
\end{IEEEproof}
 
 Lemma~\ref{lemma:R_upperbound} can be written in terms of the entropy of the type $P_{\mathscr{C}}$  as shown by the following corollary. 
 
 \begin{corollary} \label{CorRUpperBoundRelax}
Given a homogeneous $(n,M,\epsilon,B,\delta)$-code denoted by $\mathscr{C}$ for the random transformation in~\eqref{EqChannelModel} of the form in~\eqref{Eqnm_code}, the  information transmission rate $R$ in~\eqref{EqR} is such that
\begin{IEEEeqnarray}{rcl}
\nonumber
R &\leq& 
H\left( P_{\mathscr{C}} \right)  + \frac{1}{n^2} \left( \frac{1}{12} - \sum_{\ell=1}^L \frac{1}{12 P_{\mathscr{C}}(x^{(\ell)}) +1}\right)  \\
\nonumber
&& + \frac{1}{n} \left(\log \left( \sqrt{2\pi}\right)   - \sum_{\ell=1}^L \log\sqrt{2\pi P_{\mathscr{C}}(x^{(\ell)})} \right) \\
\label{EqCorRUpperBoundRelax}
&& - \frac{\log n}{n} \left( \frac{L-1}{2} \right),
\end{IEEEeqnarray}
where, $P_{\mathscr{C}}$ is the type defined in~\eqref{eq:p_bar}.
\end{corollary}
\begin{IEEEproof} 
The proof of Corollary \ref{CorRUpperBoundRelax} is presented in \cite{Zuhra-Inria-TR}.
\end{IEEEproof}

Note that all terms in~\eqref{EqCorRUpperBoundRelax}, except the entropy $H\left( P_{\mathscr{C}} \right)$, vanish with the block length $n$. This implies that the information rate is essentially constrained by the entropy of the channel input symbols. In particular, note that  $H\left( P_{\mathscr{C}}\right) \leqslant \log_2 L$.

\subsection{Information and Energy Trade-Off}

The results presented above lead to the following result. 
\begin{theorem}\label{TheoMainResult}
Given a homogeneous $(n,M,\epsilon,B,\delta)$-code denoted by $\mathscr{C}$ for the random transformation in~\eqref{EqChannelModel} of the form in~\eqref{Eqnm_code}, the following holds
\begin{IEEEeqnarray}{rcl}
\nonumber
R &\leq& 
H\left( P_{\mathscr{C}} \right)  + \frac{1}{n^2} \left( \frac{1}{12} - \sum_{\ell=1}^L \frac{1}{12 P_{\mathscr{C}}(x^{(\ell)}) +1}\right)  \\
\nonumber
&& + \frac{1}{n} \left(\log \left( \sqrt{2\pi}\right)   - \sum_{\ell=1}^L \log\sqrt{2\pi P_{\mathscr{C}}(x^{(\ell)})} \right) \\
\label{EqTheoremA}
&& - \frac{\log n}{n} \left( \frac{L-1}{2} \right);\\
\label{EqTheoremB}
B & \leq & \frac{1}{1-\delta} \sum_{\ell=1}^L P_{\mathscr{C}}(x^{(\ell)}) \mathbb{E}\left[g\left(x^{(\ell)} + W \right) \right]; \mbox{ and }\\
\nonumber
\epsilon & \geq & 1 - \exp \Big( -n H(P_{\mathscr{C}}) - n D(P_{\mathscr{C}} || Q) \\
\label{EqTheoremC}
    & & + n\log \sum_{j=1}^L\int_{\mathcal{E}_{j}} f_{Y|X}(y|x^{(j)}) \mathrm{d}y  \Big).
\end{IEEEeqnarray}
where, 
$P_{\mathscr{C}}$ is the type defined in~\eqref{eq:p_bar}; 
the pmf $Q$ is defined in~\eqref{EqQ}; 
for all $j \inCountK{L}$,  the sets $\mathcal{E}_{j}$ are defined in~\eqref{EqEell}; and
the expectation in~\eqref{EqTheoremB} is with respect to $W$, which  is a complex circularly symmetric Gaussian random variable whose real and imaginary parts have zero means and variances $\frac{1}{2}\sigma^2$.
\end{theorem}
\begin{proof}
The proof of Theorem~\ref{TheoMainResult} follows from Corollary~\ref{CorRUpperBoundRelax}, Lemma~\ref{LemmaB}; and Corollary~\ref{CorEpsilonHomogeneous}.
\end{proof}
 Theorem~\ref{TheoMainResult} quantifies the trade-off between all parameters of a homogeneous $(n,M,\epsilon,B,\delta)$-code for the random transformation in~\eqref{EqChannelModel}. Let such a code be denoted by  $\mathscr{C}$ of the form in~\eqref{Eqnm_code}. Hence, both the upper bound on the information rate in~\eqref{EqTheoremA} and the upper bound on the energy rate in~\eqref{EqTheoremB} depend on the type $P_{\mathscr{C}}$. The information rate is essentially upper bounded by the entropy of $P_{\mathscr{C}}$. Hence, codes  whose codewords are such that every channel input symbol is used the same number of times are less constrained in terms of information rate. This is the case in which  $P_{\mathscr{C}}$ is a uniform distribution.
 Alternatively, using a uniform type $P_{\mathscr{C}}$ might dramatically constrain the energy transmission rate. For instance, if the constellation is such that for at least one pair $(x_1,x_2) \in \mathcal{X}^2$, it holds that $\mathds{E}\left[g(x_1 + W)\right] < \mathds{E}\left[g(x_2 + W)\right]$ where the random variable $W$ is defined in Lemma~\ref{LemmaB}, then using the symbol $x_1$ equally often as $x_2$ certainly constraints the energy rate. Codes than might potentially exhibit the largest energy rates are those in which the symbols $x$ that maximize $\mathds{E}\left[g(x + W)\right]$ are used more often. Clearly, this deviates from the uniform distribution and thus, constraints the information rate.

Another interesting trade-off appears between the DEP and EOP. From Corollary~\ref{CorDelta}, it follows that codes whose codewords contain mainly channel input symbols $x$ that maximize $\mathds{E}\left[g(x + W)\right]$ are more reliable from the perspective of energy transmission. Alternatively, from Corollary~\ref{CorEpsilonHomogeneous2}, it follows that codes that are more reliable in terms of information transmission are those whose codewords contain more channel input symbols with the smallest DEP. That is, when $P_{\mathscr{C}}$ approaches the pmf $Q$ in~\eqref{EqQ}. 

%The following section provides an example in which these trade-offs are evidenced through a numerical analysis. 

 \section{Examples}
  \begin{figure}[t]
  \centering
  \includegraphics[width=0.45\textwidth]{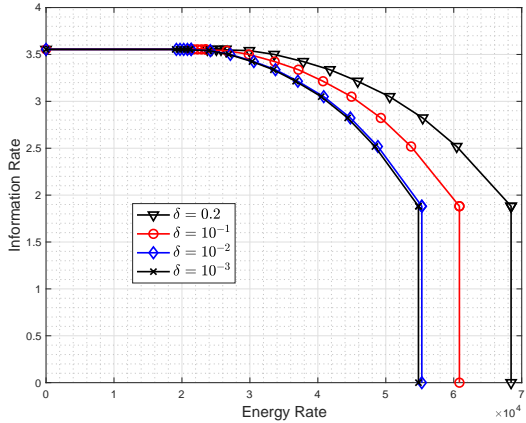}
  \caption{Upper bounds on the information rate in~\eqref{EqTheoremA} and energy rate in~\eqref{EqTheoremB} for a homogeneous code built upon a 16-QAM constellation as a function of the EOP $\delta$.}
\label{FigCapacityRegion}
\end{figure}

 \begin{figure}[t]
  \centering
  \includegraphics[width=0.45\textwidth]{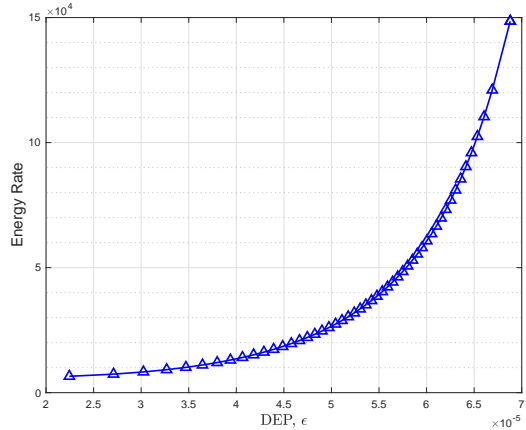}
  \caption{Upper bound on the energy rate in~\eqref{EqTheoremB} for a homogeneous code built upon a 16-QAM constellation as a function of the lower bound on the DEP in~\eqref{EqTheoremC}. }
\label{FigRateEnergyDEP}
\end{figure}

Consider a homogeneous $(n,M,\epsilon,B,\delta)$-code $\mathscr{C}$ (Definition~\ref{DefHC}) for the random transformation in~\eqref{EqChannelModel} of the form in~\eqref{Eqnm_code} built upon a $16$-QAM constellation. More specifically, the set of channel input symbols $\mathcal{X}$ in~\eqref{EqCIsymbols} is: $\mathcal{X} = \lbrace (5+\mathrm{i}5), (-5+\mathrm{i}5), (-5-\mathrm{i}5), (5-\mathrm{i}5), (15+\mathrm{i}5), (5+\mathrm{i}15), (-5+\mathrm{i}15), (-15+\mathrm{i}5), (-15-\mathrm{i}5), (-5-\mathrm{i}15),  (5-\mathrm{i}15), (15-\mathrm{i}5), (15+\mathrm{i}15), (-15+\mathrm{i}15), (-15-\mathrm{i}15), (15-\mathrm{i}15) \rbrace.$
The decoding sets $\mathcal{D}_1$, $\mathcal{D}_2$, $\ldots$, $\mathcal{D}_{16}$  in~\eqref{EqDitHC} are those obtained from the maximum-\`{a}-posteriori (MAP) estimator \cite{proakis} given the type $P_{\mathscr{C}}$ in~\eqref{eq:p_bar}. The block-length is fixed at $n = 80$ channel uses and the noise variance $\sigma^2 = 2$ in~\eqref{EqDensities}.
The energy harvested at the EH is calculated taking into account the non-linearities at the receiver described in \cite{varasteh2017wireless}. In particular, given a type $P_{\mathscr{C}}$, the expression  $\sum_{\ell=1}^L P_{\mathscr{C}}(x^{(\ell)}) \mathbb{E}\left[g\left(x^{(\ell)} + W \right) \right]$ in~\eqref{EqTheoremB} is obtained by using \cite[Proposition $1$]{varasteh2017wireless}, which is denoted by $P_{del}$ therein. The choice of parameters while using  \cite[Proposition $1$]{varasteh2017wireless} is $k_2 = k_4 = 1$ and $|h| = |\tilde{h}| = 1$. 

\noindent
In Figure~\ref{FigCapacityRegion}, given a type $P_{\mathscr{C}}$, the upper bound on the information transmission rate and the upper bound on the energy transmission rate in~\eqref{EqTheoremA} and~\eqref{EqTheoremB}, respectively, are depicted. Each point in the curves corresponds to a different type $P_{\mathscr{C}}$. Therein, it becomes evident that smaller EOPs imply smaller information and energy transmission rates.
The maximum information rate achievable with a $16$-QAM constellation is four bits per channel use. Nonetheless, when codewords are forced to exhibit the same type (homogeneous code), the information transmission rate does not exceed $3.51$ bits per channel use. This maximum holds when the energy transmission rate is smaller than $2 \times 10^4$ energy units, which corresponds to the maximum energy that might be delivered by codes whose type $P_{\mathscr{C}}$ approaches the uniform distribution on $\mathcal{X}$. Beyond this value, increasing the energy transmission rate implies deviating from the uniform distribution, which penalises the information transmission rate. The maximum energy rate implies types $P_{\mathscr{C}}$ that concentrate on the symbols $(15+\mathrm{i}15)$, $(-15+\mathrm{i}15)$, $(-15-\mathrm{i}15)$, $(15-\mathrm{i}15)$. A type $P_{\mathscr{C}}$ that approaches the uniform distribution over these four channel inputs induces a lower bound on the information transmission rate of $1.9$ bits per channel use. A type $P_{\mathscr{C}}$ that concentrates on one of these channel input symbols induces a zero information rate. 

\noindent
Figure~\ref{FigRateEnergyDEP} illustrates the upper bound on the energy transmission rate in~\eqref{EqTheoremB} as a function of the lower bound on the DEP in~\eqref{EqTheoremC} when the EOP is kept constant at  $\delta = 10^{-4}$ and $n =80$. More specifically, each point on the curve in Figure~\ref{FigRateEnergyDEP} for a given value of the DEP is associated with a type $P_{\mathscr{C}}$. This type is used to determine the upper bound on the energy transmission rate, which leads to the curve in Figure~\ref{FigRateEnergyDEP}.
Higher DEPs require types $P_{\mathscr{C}}$ with less constraints, which increases the bound on the energy rate.  

\balance
 \bibliographystyle{IEEEtran}
\bibliography{ITW2021}
\balance
 \end{document}